\documentclass{easychair}
\usepackage[utf8]{inputenc}
\usepackage[T1]{fontenc}

\usepackage{enumerate}

\usepackage{bussproofs}

\usepackage{tikz}
\usetikzlibrary{matrix}

\usepackage{amsfonts}
\usepackage{amsmath}
\usepackage{amssymb}
\usepackage{amsthm}
\usepackage{csquotes}

\theoremstyle{plain}
\newtheorem{definition}{Definition}
\newtheorem{theorem}{Theorem}
\newtheorem{lemma}{Lemma}

\newtheorem{problem}{Problem}
\newtheorem{example}{Example}

\def\AA{\mathbb{A}}

\def\NN{\mathbb{N}}

\DeclareMathOperator{\rank}{rank}
\DeclareMathOperator{\order}{order}

\def\AX{\ensuremath{(\text{Ax})}}

\def\AI{\ensuremath{(\text{$\to$I})}}
\def\AE{\ensuremath{(\text{$\to$E})}}
\def\II{\ensuremath{(\text{$\cap$I})}}
\def\IE{\ensuremath{(\text{$\cap$E})}}

\title{Rank 3 Inhabitation of Intersection Types Revisited (Extended Version)\footnote{Based on: TYPES 2016, Types for Proofs and Programs -- 22nd Meeting, Novi Sad, Serbia, 2016, pp. 29--30.}}

\author{
Andrej Dudenhefner\inst{1}
\and
Jakob Rehof\inst{1}
}

\institute{
  Technical University of Dortmund,
  Dortmund, Germany\\
  \email{\{andrej.dudenhefner, jakob.rehof\}@cs.tu-dortmund.de}
 }

\authorrunning{Dudenhefner and Rehof}

\titlerunning{Rank 3 Inhabitation Revisited}

\begin{document}
\maketitle
\begin{abstract}
We revisit the undecidability result of rank 3 intersection type inhabitation (Urzyczyn 2009) in pursuit of two goals. First, we strengthen the previous result by showing that intersection type inhabitation is undecidable for types of rank $3$ \emph{and order 3}, i.e. it is not necessary to introduce new functional dependencies (new instructions) during proof search. Second, we pinpoint the principles necessary to simulate Turing machine computation \emph{directly}, whereas previous constructions used a highly parallel and non-deterministic computation model. Since our construction is more concise than existing approaches taking no detours, we believe that it is valuable for a better understanding of the expressiveness of intersection type inhabitation.
\end{abstract}

\section{Introduction}
Intersection types capture deep semantic properties of $\lambda$-terms such as normalization and were subject to extensive study for decades~\cite{BDS13}. 
  Due to their expressive power, intersection type inhabitation (given a type, does there exist a term having the type?) is undecidable in standard intersection type systems~\cite{ghilezan96}. 
  
  Urzyczyn showed in 2009~\cite{urzy09} that restricted to rank 2 intersection type inhabitation is exponential space complete and is undecidable for rank 3 and up. Later, Salvati et al. discovered the equivalence of intersection type inhabitation and $\lambda$-definability~\cite{smgb12}, which further improved our understanding of the expressiveness of intersection types.
  
  One can take several routes in order to show undecidability of intersection type inhabitation (abbreviated by IHP, resp. IHP3 for rank 3). In~\cite{BDS13} the following reduction is performed: EQA $\leq$ ETW $\leq$ WTG $\leq$ IHP, where EQA is the emptiness problem for queue automata, ETW is the emptiness problem for typewriter automata and WTG is the problem of determining whether one can win a tree game. A different route taken in~\cite{urzy09} performs the following reduction: ELBA $\leq$ SSTS1 $\leq$ HETM $\leq$ IHP3, where ELBA is the emptiness problem for linear bounded automata, SSTS1 is the problem of deciding whether there is a word that can be rewritten to 1s in a simple semi-Thue system and HETM is the halting problem for expanding tape machines. Alternatively, following the route of \cite{smgb12} via semantics, the following reduction can be performed: WSTS $\leq$ LDF $\leq$ IHP3, where WSTS is the word problem in semi-Thue systems and LDF is the $\lambda$-definability problem. Each of these routes introduces its own machinery that may distract from the initial question. As a result, it is challenging to even give examples of particularly hard inhabitation problem instances, or distinguish necessary properties on a finer scale than the rank restriction.
  
  In this work (cf. Section \ref{sec:tm}), we reduce the halting problem for Turing machines to intersection type inhabitation directly. In particular, we show how arbitrary Turing machine computations can be directly simulated using proof search. Additionally (Section \ref{sec:ssts}), we outline the simplest construction (to date) to show undecidability of rank 3 inhabitation using simple semi-Thue systems. The main benefit of presented approaches is their corresponding accessibility and simplicity. 
  

\section{Preliminaries}

We briefly assemble the necessary prerequisites in order to discuss inhabitation in the intersection type system. $\lambda$-terms (cf. Definition \ref{def:lambda-terms}) are denoted by $M,N,L$ and intersection types (cf. Definition \ref{def:simple-types}) are denoted by $\sigma, \tau$, while type atoms are denoted by $a,b,c$ and drawn from the denumerable set $\AA$. The intersection operator $\cap$ binds stronger than $\to$ and is assumed to be idempotent, commutative and associative.
The rules of the considered intersection type system (also used in \cite{urzy09}) are given in Definition \ref{def:type-system}.

\begin{definition}[$\lambda$-Terms]
\label{def:lambda-terms}
$M,N,L ::= x \mid (\lambda x.M) \mid (M \, N)$
\end{definition}

\begin{definition}[Intersection Types]
\label{def:simple-types}
$\sigma, \tau ::= a \mid \sigma \to \tau \mid \sigma \cap \tau \text{ where } a \in \AA$
\end{definition}

\begin{definition}[Intersection Type System]\ \\
\label{def:type-system}

\medskip

\begin{tabular}{c c}
{\RightLabel{\AX}
\AxiomC{}
\UnaryInfC{$\Gamma, x : \sigma \vdash x : \sigma $}
\DisplayProof} \\
\\
{\RightLabel{\AI}
\AxiomC{$\Gamma, x: \sigma \vdash M : \tau$}
\UnaryInfC{$\Gamma \vdash \lambda x.M : \sigma \to \tau$}
\DisplayProof} & 
{\RightLabel{\AE}
\AxiomC{$\Gamma \vdash M : \sigma \to \tau$}
\AxiomC{$\Gamma \vdash N : \sigma$}
\BinaryInfC{$\Gamma \vdash M \; N : \tau$}
\DisplayProof}\\
\\
{\RightLabel{\II}
\AxiomC{$\Gamma \vdash M : \sigma$}
\AxiomC{$\Gamma \vdash M : \tau$}
\BinaryInfC{$\Gamma \vdash M : \sigma \cap \tau$}
\DisplayProof} & 
{\RightLabel{\IE}
\AxiomC{$\Gamma \vdash M : \sigma \cap \tau$}
\UnaryInfC{$\Gamma \vdash M : \sigma$}
\DisplayProof}
\end{tabular}
\end{definition}

\noindent 
After Leivant \cite{le83} we define the \emph{rank} and the \emph{order} of a type.

\begin{definition}[Rank]
\begin{align*}
\rank(\tau) & = 0 \text{ if } \tau \text{ is a simple type, i.e. containing no } \cap\\
\rank(\sigma \cap \tau) & = \max\{1, \rank(\sigma), \rank(\tau)\}  \\
\rank(\sigma \to \tau) & = \max\{1+\rank(\sigma), \rank(\tau)\} 
\end{align*}
\end{definition}

\begin{definition}[Order]
\begin{align*}
\order(a) & = 0 \\
\order(\sigma \cap \tau) & = \max\{\order(\sigma), \order(\tau)\}  \\
\order(\sigma \to \tau) & = \max\{1+\order(\sigma), \order(\tau)\} 
\end{align*}
\end{definition}

\begin{example}
Let $\tau = \bigcap_{i=1}^n (((d_i \to c_i) \to b_i) \to a_i) \to a$. We have $\rank(\tau) = 2$ and $\order(\tau)=4$.
Note that types that are similar to $\tau$ are used in \cite{urzy09} to encode expanding tape machine switches.
\end{example}

\begin{problem}[Intersection Type Inhabitation]
\label{prb:inhabitation}
Given a type $\tau$, is there a $\lambda$-term $M$ such that $\vdash M : \tau$?
\end{problem}

\newpage

\def\sp{\text{\textvisiblespace}}
\def\tup#1#2{\langle #1, #2 \rangle}

\section{Turing Machine Simulation}
\label{sec:tm}
The undecidability proof in \cite{urzy09} reduces emptiness of linear bounded automata to a word existence problem in simple semi-Thue systemd to halting of expanding tape machines to inhabitation of an intersection type $\tau$ with $\rank(\tau)=3$ and $\order(\tau)=4$. In \cite{urzy09} order of $4$ is necessary to hide and restore a word on the tape of the machine by introducing new instructions. In this section, we reduce the halting problem to inhabitation of an intersection type $\tau_\star$ with $\rank(\tau_\star)=\order(\tau_\star)=3$.
For this purpose, fix a Turing machine $T = (\Sigma, Q, q_0, q_f, \delta)$, where 
\begin{itemize}
\item $\Sigma$ is the finite, non-empty set of tape symbols
\item $Q$ is the finite, non-empty set of states
\item $q_0 \in Q$ is the initial state
\item $q_f \in Q$ is the final state
\item $\delta : (Q \setminus \{q_f\}) \times \Sigma \to Q \times \Sigma \times \{+1,-1\}$ is the transition function
\end{itemize}
Let $\dot{\cup}$ denote disjoint union. We define the set of type atoms 
$$\AA = \Sigma \mathbin{\dot{\cup}} \{l, r, \bullet\} \mathbin{\dot{\cup}} \{\tup{q}{a} \mid q \in Q, a \in \Sigma\} \mathbin{\dot{\cup}} \{\circ, *, \#, \$ \}$$
In the first part of our construction, we capture computation of $T$ using the following types: 
\begin{align*}
\sigma_f = & \bigcap\limits_{c \in \Sigma} \tup{q_f}{c} \cap \bigcap\limits_{c \in \Sigma} c \\
\text{ for each } t = & ((q,c) \mapsto (q', c', +1)) \in \delta\\
\sigma_t = & \bigcap\limits_{a \in \Sigma} (\bullet \to a \to a) \cap (l \to c' \to \tup{q}{c}) \cap \bigcap\limits_{a \in \Sigma} (r \to \tup{q'}{a} \to a) \\
\text{ for each } t = & ((q,c) \mapsto (q', c', -1)) \in \delta\\
\sigma_t = &\bigcap\limits_{a \in \Sigma} (\bullet \to a \to a) \cap (r \to c' \to \tup{q}{c}) \cap \bigcap\limits_{a \in \Sigma} (l \to \tup{q'}{a} \to a) \\
\end{align*}
Intuitively, the defined types are used to achieve following goals
\begin{itemize}
\item $\tup{q}{a}$ represents that $T$ is at the position containing $a$ in state $q$
\item $l,r$ (left, right) mark neighboring cells while other cells are marked by $\bullet$
\item $\sigma_f$ recognizes whether an accepting state is reached
\item $\sigma_t$ transforms the state according to $\delta$
\end{itemize}

\noindent Let us abbreviate $[i,j] := \{i, i+1, \ldots, j\}$.

\ifx\false

\begin{lemma}
We have
\begin{enumerate}[$(i)$]
\item $\Gamma_1 \vdash e : \tau_1, \ldots, \Gamma_{n-1} \vdash e : \tau_{n-1}, \Gamma_n \vdash e : \tau_n \cap \sigma_n$ iff \\
$\Gamma_1 \vdash e : \tau_1, \ldots, \Gamma_n \vdash e : \tau_n, \Gamma_n \vdash e : \sigma_n$
\item $\Gamma_1 \vdash e : \sigma_1 \to \tau_1, \ldots, \Gamma_n \vdash e : \sigma_n \to \tau_n$ iff \\
$\Gamma_1, x : \sigma_1 \vdash e' : \tau_1, \ldots, \Gamma_n, x : \sigma_n \vdash e' : \tau_n$
\item $\Gamma_1 \vdash e : a_1, \ldots, \Gamma_n \vdash e : a_n$ iff \\
$\Gamma_1 \vdash x : \sigma_1^1 \to \ldots \to \sigma_1^l \to a_1, \ldots, \Gamma_n \vdash x : \sigma_n^1 \to \ldots \to \sigma_n^l \to a_n$ and 
$\Gamma_1 \vdash e_i : \sigma_1^i, \ldots, \Gamma_n \vdash e_i : \sigma_n^i$ for $1 \leq i \leq l$
\end{enumerate}
\end{lemma}

\fi

\begin{lemma}
\label{lem:rank3_inner_equivalence}
Given $q \in Q$, $s = a_1 \ldots a_n \in \Sigma^n$ and $p \in [1,n]$ we have
$(q, p, s)$ is accepting in $T$ using at most $n$ tape cells iff there exists a term $e$ such that $\Gamma_p \vdash e : \tup{q}{a_p}$ and $\Gamma_i \vdash e : a_i$ for $i \in [1,n] \setminus \{p\}$, where
\begin{itemize}
\item $\Gamma = \{x_f : \sigma_f\} \cup \{x_t : \sigma_t \mid t \in \delta\} $
\item $\Gamma_1 = \Gamma \cup \{y_1 : l\} \cup \{y_j : \bullet \mid j \in [2,n-1]\}$
\item $\Gamma_i = \Gamma \cup \{y_{i-1} : r, y_i : l\} \cup \{y_j : \bullet \mid j \in [1,n-1] \setminus \{i-1,i\}\}$ for $i \in [2,n-1]$
\item $\Gamma_n = \Gamma \cup \{y_{n-1} : r\} \cup \{y_j : \bullet \mid j \in [1,n-2]\}$
\end{itemize}
\end{lemma}

\begin{proof}
By construction $(y_i : l) \in \Gamma_i$ and $(y_i : r) \in \Gamma_{i+1}$ for $i \in [1,n-1]$. Intuitively, this property is used to identify \enquote{neighboring} contexts.

{\bf \enquote{$\Rightarrow$}:}
Assume $(q, p, a_1 \ldots  a_n)$ is accepting in $T$ using at most $n$ tape cells. If $q = q_f$, then $\Gamma_i \vdash x_f : a_i$ for $i \in [1,n] \setminus \{p\}$ and $\Gamma_p \vdash x_f : \tup{q}{a_p}$.
To capture computation, we represent the transition $C \stackrel{t \in \delta}{\leadsto} C'$ with $C = (q, p, a_1 \ldots a_p \ldots a_n)$ and $C' = (q', p+1, a_1 \ldots a_{p-1}, c', a_{p+1} \ldots a_n)$ by the term $M_C^{C'} = x_t\,y_p$ with 
\begin{align*}
&\Gamma_i \vdash M_C^{C'} : a_i \to a_i \text{ for } i \in [1,n] \setminus \{p,p+1\} \\
&\text{and } \Gamma_p \vdash M_C^{C'} : c' \to \tup{q}{a_p} \text{ and } \Gamma_{p+1} \vdash M_C^{C'} : \tup{q'}{a_{p+1}} \to a_{p+1}
\end{align*}
Similarly, we represent the transition $C \stackrel{t \in \delta}{\leadsto} C'$ with $C = (q, p, a_1 \ldots a_p \ldots a_n)$ and $C' = (q', p-1, a_1 \ldots a_{p-1}, c', a_{p+1} \ldots a_n)$ by the term $M_C^{C'} = x_t\,y_{p-1}$ with 
\begin{align*}
&\Gamma_i \vdash M_C^{C'} : a_i \to a_i \text{ for } i \in [1,n] \setminus \{p-1,p\} \\
&\text{and } \Gamma_p \vdash M_C^{C'} : c' \to \tup{q}{a_p} \text{ and } \Gamma_{p-1} \vdash M_C^{C'} : \tup{q'}{a_{p-1}} \to a_{p-1}
\end{align*}
Given a sequence of configurations $(q,p,s) = C_1 \leadsto \ldots \leadsto C_m = (q_f, p', s')$ for some $p' \leq n$ and $s' \in \Sigma^n$ let $M = M_{C_{1}}^{C_2} ( M_{C_{2}}^{C_{3}} ( \ldots ( M_{C_{m-1}}^{C_{m}} \, x_f ) \ldots ))$. By induction on $m$ we have $\Gamma_i \vdash M : a_i$ for $i \in [1,n] \setminus \{p\}$ and $\Gamma_p \vdash M : \tup{q}{a_p}$.

{\bf \enquote{$\Leftarrow$}:}
Assume $\Gamma_i \vdash M : a_i$ for $i \in [1,n] \setminus \{p\}$ and $\Gamma_p \vdash M : \tup{q}{a_p}$ for some term $M$ in beta normal form (cf.~\cite{ghilezan96}). We show that $(q,p,s)$ is accepting by induction on the structure of $M$.
Due to $\Gamma_p \vdash M : \tup{q}{a_p}$, we have three possible cases:
\begin{description}
\item[Case 1:] $M = x_f$. Therefore, $q = q_f$ and $(q,p,s)$ is accepting.
\item[Case 2:] $M = x_t\,L\,N$ with $t = (q,a_p) \mapsto (q', c', +1)$. We have
\begin{align*}
& \Gamma_p \vdash M : \tup{q}{a_p}\\
\implies & \Gamma_p \vdash L : l \\
\implies & L = y_p \text{ and } 1 \leq p < n \\
\implies & \Gamma_{p+1} \vdash L : r \text{ and } \Gamma_i \vdash L : \bullet \text{ for } i \in [1,n] \setminus \{p,p+1\}\\
\implies & \Gamma_i \vdash N : a_i \text{ for } i \in [1,n] \setminus \{p,p+1\} \\
& \text{and } \Gamma_p \vdash N : c' \text{ and } \Gamma_{p+1} \vdash N : \tup{q'}{a_{p+1}}
\end{align*}
By the induction hypothesis $(q', p+1, a_1 \ldots a_{p-1} c' a_{p+1} \ldots a_n)$ is accepting in $T$ using at most $n$ tape cells. Therefore, $(q,p,s)$ is accepting in $T$ using at most $n$ tape cells.
\item[Case 3:] $M = x_t\,L\,N$ with $t = (q,a_p) \mapsto (q', c', -1)$. We have
\begin{align*}
& \Gamma_p \vdash M : \tup{q}{a_p}\\
\implies & \Gamma_p \vdash L : r \\
\implies & L = y_{p-1} \text{ and } 1 < p \leq n \\
\implies & \Gamma_{p-1} \vdash L : l \text{ and } \Gamma_i \vdash L : \bullet \text{ for } i \in [1,n] \setminus \{p-1,p\}\\
\implies & \Gamma_i \vdash N : a_i \text{ for } i \in [1,n] \setminus \{p-1,p\}\\
& \text{and } \Gamma_p \vdash N : c' \text{ and } \Gamma_{p-1} \vdash N : \tup{q'}{a_{p-1}}\\
\end{align*}
By the induction hypothesis $(q', p-1, a_1 \ldots a_{p-1} c' a_{p+1} \ldots a_n)$ is accepting in $T$ using at most $n$ tape cells. Therefore, $(q,p,s)$ is accepting in $T$ using at most $n$ tape cells.
\end{description}
\end{proof}

\newpage

In the second part of our construction we create the particular inhabitation problem to apply Lemma \ref{lem:rank3_inner_equivalence}. The main idea is to use a technique similar to the encoding of expansion in expanding tape machines~\cite{urzy09}. However, instead of introducing new functional dependencies (new machine instructions in the sense of~\cite{urzy09}) it is sufficient to introduce type atoms that link neighboring judgements. We define the following types:
\begin{align*}
\sigma_* = &((\bullet \to \circ) \to \circ) \cap ((\bullet \to *) \to *) \cap ((l \to *) \to \#) \cap ((r \to \#) \cap (\bullet \to \$) \to \$) \\
\sigma_0 = &((\bullet \to \tup{q_0}{\sp}) \to \circ) \cap ((\bullet \to \sp) \to *) \cap ((l \to \sp) \to \#) \cap ((r \to \sp) \to \$)\\
\tau_\star  = & \sigma_0 \to \sigma_* \to \sigma_f \to \sigma_{t_1} \to \ldots \to \sigma_{t_k} \to (l \to \circ) \cap (r \to \#) \cap (\bullet \to \$) \\
& \text{where } \delta = \{t_1, \ldots, t_k\}\\
\end{align*}
Note that $\sp \in \Sigma$ is the space symbol. Again, let us develop an intuition for the defined types:
\begin{itemize}
\item $\circ$ marks the first, $\$$ the last, $\#$ the second last and $*$ all other cells during tape expansion
\item $\sigma_*$ is used for tape expansion at its end
\item $\sigma_0$ initializes $T$ to the state $q_0$ and the empty tape replacing $\circ, *, \#, \$$ marks 
\end{itemize}
Observe that $\rank(\tau_\star) = \order(\tau_\star) = 3$. 

\begin{lemma}
\label{lem:rank3_outer_equivalence}
There exists a term $M$ such that $\vdash M : \tau_\star$ iff 
there exists an $n \in \NN$ and a term $N$ such that $\Gamma_1 \vdash N : \tup{q_0}{\sp}$ and $\Gamma_i \vdash N : \sp$ for $i \in [2,n]$ where \begin{itemize}
\item $\Gamma = \{x_* : \sigma_*, x_0 : \sigma_0, x_f : \sigma_f\} \cup \{x_t : \sigma_t \mid t \in \delta\} $
\item $\Gamma_1 = \Gamma \cup \{y_1 : l\} \cup \{y_j : \bullet \mid j \in [2,n-1]\}$
\item $\Gamma_i = \Gamma \cup \{y_{i-1} : r, y_i : l\} \cup \{y_j : \bullet \mid j \in [1,n-1] \setminus \{i-1,i\}\}$ for $i \in [2,n-1]$
\item $\Gamma_n = \Gamma \cup \{y_{n-1} : r\} \cup \{y_j : \bullet \mid j \in [1,n-2]\}$
\end{itemize}
\end{lemma}

\begin{proof}
To improve accessibility of the proof, we visualize necessary/sufficient steps (according to the inhabitation algorithm given in \cite{urzy09}) to prove $\vdash M : \tau_\star$. A node $\sigma$ states that $\sigma$ is inhabited in the context described by the union of edge labels along the path from $\sigma$ to the root. Additionally, at each level of the tree the shape of inhabitants (which have to be identical due to the nature of $\cap$) is ascribed on the left. The inhabitant is therefore the composition of the ascriptions. 

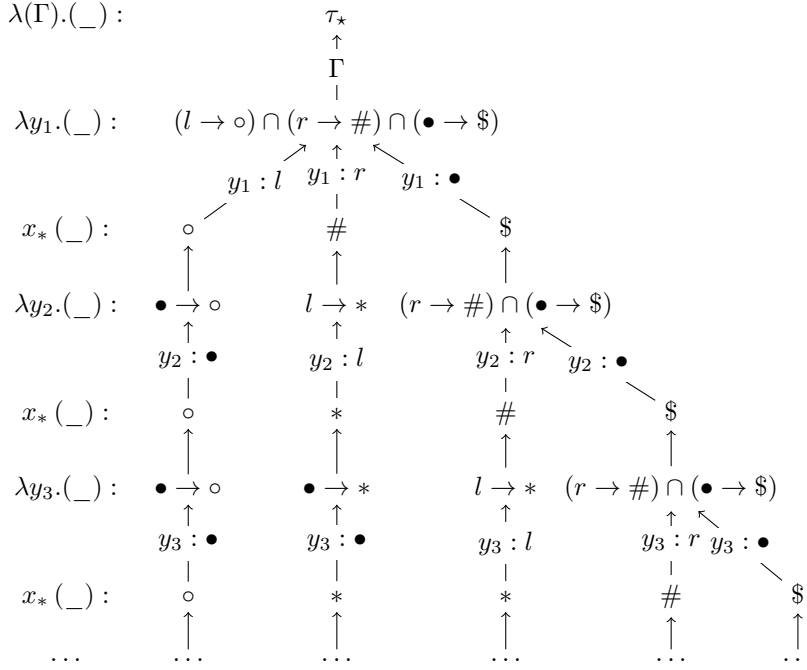
\begin{tikzpicture}
  \matrix (m) [matrix of math nodes,row sep=1.2em,column sep=0.5em,minimum width=0.5em]
  {
   \lambda (\Gamma).(\_) : & & \tau_\star \\ \\
   \lambda y_1.(\_) : &  & \hspace{-3em} (l \to \circ) \cap (r \to \#) \cap (\bullet \to \$) \hspace{-3em} \\ \\
   x_*\,(\_) : & \circ & \# & \$  & \\ 
   \lambda y_2.(\_) : & \bullet \to \circ & l \to * & \hspace{-2em} (r \to \#) \cap (\bullet \to \$) \hspace{-2em}\\ \\
   x_*\,(\_) : & \circ & * & \# & \$\\ 
   \lambda y_3.(\_) : & \bullet \to \circ & \bullet \to * & l \to * & \hspace{-1em} (r \to \#) \cap (\bullet \to \$) \hspace{-1em}\\ \\
   x_*\,(\_) : & \circ & * & * & \# & \$\\ 
   \ldots & \ldots & \ldots & \ldots & \ldots & \ldots\\
    };
  \path
    (m-3-3) edge[->] node[fill=white] {$\Gamma$} (m-1-3)
    (m-5-2) edge[->] node[fill=white] {$y_1 : l$} (m-3-3)
    (m-5-3) edge[->] node[fill=white] {$y_1 : r$} (m-3-3)
    (m-5-4) edge[->] node[fill=white] {$y_1 : \bullet$} (m-3-3)
    (m-6-2) edge[->]  (m-5-2)
    (m-6-3) edge[->]  (m-5-3)
    (m-6-4) edge[->]  (m-5-4)
    (m-8-2) edge[->] node[fill=white] {$y_2 : \bullet$} (m-6-2)
    (m-8-3) edge[->] node[fill=white] {$y_2 : l$} (m-6-3)
    (m-8-4) edge[->] node[fill=white] {$y_2 : r$} (m-6-4)
    (m-8-5) edge[->] node[fill=white] {$y_2 : \bullet$} (m-6-4)
    (m-9-2) edge[->]  (m-8-2)
    (m-9-3) edge[->]  (m-8-3)
    (m-9-4) edge[->]  (m-8-4)
    (m-9-5) edge[->]  (m-8-5)
    (m-11-2) edge[->] node[fill=white] {$y_3 : \bullet$} (m-9-2)
    (m-11-3) edge[->] node[fill=white] {$y_3 : \bullet$} (m-9-3)
    (m-11-4) edge[->] node[fill=white] {$y_3 : l$} (m-9-4)
    (m-11-5) edge[->] node[fill=white] {$y_3 : r$} (m-9-5)
    (m-11-6) edge[->] node[fill=white] {$y_3 : \bullet$} (m-9-5)
    (m-12-2) edge[->]  (m-11-2)
    (m-12-3) edge[->]  (m-11-3)
    (m-12-4) edge[->]  (m-11-4)
    (m-12-5) edge[->]  (m-11-5)
    (m-12-6) edge[->]  (m-11-6)
    ;
\end{tikzpicture}

First, the variables in $\Gamma$ and $y_1$ are abstracted. Next follows an arbitrary number of $x_*$ applications in combination with an abstraction, that increments the current width of the tree. Intuitively, $\circ$ is used to mark the first context; $\$$ is used to mark the last context; $\#$ is used to introduce consecutive $y_i : r$ and $y_{i+1} : l$.

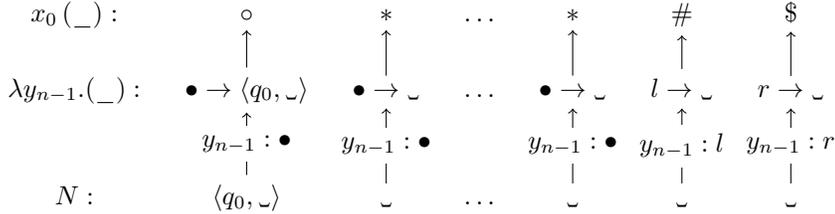
\begin{tikzpicture}
  \matrix (m) [matrix of math nodes,row sep=1.2em,column sep=1em,minimum width=1em]
  {
   x_0\,(\_) : & \circ & * & \ldots & * & \# & \$ \\
\lambda y_{n-1}.(\_) : & \bullet \to \tup{q_0}{\sp} & \bullet \to \sp & \ldots & \bullet \to \sp & l \to \sp & r \to \sp\\ \\
 N : & \tup{q_0}{\sp} & \sp & \ldots & \sp & \sp & \sp\\
    };
  \path
    (m-2-2) edge[->]  (m-1-2)
    (m-2-3) edge[->]  (m-1-3)
    (m-2-5) edge[->]  (m-1-5)
    (m-2-6) edge[->]  (m-1-6)
    (m-2-7) edge[->]  (m-1-7)
    (m-4-2) edge[->] node[fill=white] {$y_{n-1} : \bullet$} (m-2-2)
    (m-4-3) edge[->] node[fill=white] {$y_{n-1} : \bullet$} (m-2-3)
    (m-4-5) edge[->] node[fill=white] {$y_{n-1} : \bullet$} (m-2-5)
    (m-4-6) edge[->] node[fill=white] {$y_{n-1} : l$} (m-2-6)
    (m-4-7) edge[->] node[fill=white] {$y_{n-1} : r$} (m-2-7)
    ;
\end{tikzpicture}

Finally, an application of $x_0$ in combination with one abstraction of $y_{n-1}$, where $n$ is the current width of the tree, produces exactly $\Gamma_1 \vdash N : \tup{q_0}{\sp}$ and $\Gamma_i \vdash N : \sp$ for $i \in [2,n]$, that can be verified inductively reading the edge labels along the path to the root.

Using the presented tree, we sketch the proof as follows.

{\bf \enquote{$\Rightarrow$}:} If $\tau_\star$ is inhabited, then we can read the tree from the root to the leaves to get necessary conditions $\Gamma_1 \vdash N : \tup{q_0}{\sp}$ and $\Gamma_i \vdash N : \sp$ for $i \in [2,n]$ for some term $N$.

{\bf \enquote{$\Leftarrow$}:} If $\Gamma_1 \vdash N : \tup{q_0}{\sp}$ and $\Gamma_i \vdash N : \sp$ for $i \in [2,n]$, then we can read the tree from leaves to the root to get the sufficient condition, namely the inhabitant $M$, such that $\vdash M : \tau_\star$.
\end{proof}

\begin{theorem}
Intersection type inhabitation (Problem \ref{prb:inhabitation}) is undecidable in rank $3$ and order $3$.
\end{theorem}

\begin{proof}
Adding $x_* : \sigma_*$ and $x_0 : \sigma_0$ to all contexts does not invalidate Lemma \ref{lem:rank3_inner_equivalence} because $\sigma_*$ and $\sigma_0$ work on the special atoms $\circ, *, \#$ and $\$$.
Therefore, by Lemma \ref{lem:rank3_outer_equivalence} and Lemma \ref{lem:rank3_inner_equivalence} we have that $\tau_\star$ is inhabited iff $T$ accepts the empty word.
\end{proof}

\section{Simple Semi-Thue System Simulation}
\label{sec:ssts}
The undecidability result in Section \ref{sec:tm} is of didactic interest because it shows how Turing machine computation can be simulated by proof search directly. If one is only interested in minimal requirements for undecidability, then it is helpful to consider simple semi-Thue systems.

\begin{definition}[Simple Semi-Thue System]
A semi-Thue system over an alphabet $\Sigma$, where each rule has the form $ab \Rightarrow  cd$, for some $a,b,c,d \in \Sigma$, is called a \emph{simple semi-Thue system}.
\end{definition}

\begin{problem}
\label{prb:ssts01}
Given a simple semi-Thue system $R$, is there an $n \in \NN$ such that $0^n \twoheadrightarrow_R 1^n$?
\end{problem}

\begin{lemma}
\label{lem:ssts-word}
Problem \ref{prb:ssts01} is undecidable.
\end{lemma}

\begin{proof}
Reduction from the empty word halting problem for Turing machines. Turing machine computation requires at each step the modification of two neighboring cells. The unknown parameter $n$ corresponds to the required tape length.
\end{proof}

Given a simple semi-Thue system $R$ over an alphabet $\Sigma$ with $\bullet, *, \#, \$, l, r \not\in \Sigma$ we (similarly to Section \ref{sec:tm}) construct the following types of rank and order at most $3$:
\begin{align*}
\sigma_* & = ((\bullet \to *) \to *) \cap ((l \to *) \to \#) \cap ((r \to \#) \cap (\bullet \to \$) \to \$) \\
\sigma_0 & = ((\bullet \to 0) \to *) \cap ((l \to 0) \to \#) \cap ((r \to 0) \to \$)\\
\sigma_1 &= 1\\
\sigma_{ab \Rightarrow cd} & = (l \to c \to a) \cap (r \to d \to b) \cap \bigcap\limits_{e \in \Sigma} (\bullet \to e \to e) \text{ for each } ab \Rightarrow cd \in R \\
\tau_\star & = \sigma_0 \to \sigma_* \to \sigma_1 \to \sigma_{t_1} \to \ldots \to \sigma_{t_k} \to (l \to *) \cap (r \to \#) \cap (\bullet \to \$) \\
& \text{where } R = \{t_1, \ldots, t_k\}
\end{align*}

\begin{lemma}
The type $\tau_\star$ is inhabited iff there exists an $n \in \NN$ such that $0^n \twoheadrightarrow_R 1^n$.
\end{lemma}

\begin{proof}
The argumentation is the same as in the proofs of Lemma \ref{lem:rank3_outer_equivalence} and Lemma \ref{lem:rank3_inner_equivalence}. The type $\sigma_*$ is used for initial expansion; $\sigma_0$ is used for initialization; $\sigma_t$ for transition for $t \in R$; and $\sigma_1$ for termination.
\end{proof}

\bibliographystyle{plainurl}
\bibliography{bibliographyLS14}
\end{document}